%% file: Diagonal.tex
\documentclass[%
 reprint,
 amsmath,amssymb,
 aps,
]{revtex4-2}
\usepackage{graphicx}
\usepackage{dcolumn}
\usepackage{bm}

\usepackage{mathtools}
\usepackage{amsmath,amssymb,amsthm}
\newtheorem{definition}{Definition}
\newtheorem{theorem}{Theorem}
\newtheorem{lemma}{Lemma}
\newtheorem*{thmA}{Theorem A}
\newtheorem*{thmB}{Theorem B}
\newtheorem*{propertyA}{Property A}
\newtheorem{obs}{Observation}
\newtheorem{corollary}{Corollary}
\newtheorem{remark}{Remark}
\usepackage{enumerate}
\begin{document}

\preprint{APS/123-QED}

\title{The Diagonal distance of Codeword Stabilized Codes}

\author{Upendra S. Kapshikar}
 \email{e0382999@u.nus.edu}
\affiliation{%
 Centre for Quantum Technologies, \\
 National University of Singapore. 
}%




\date{\today}

\begin{abstract}
{Quantum degeneracy in error correction is a feature unique to quantum error correcting codes, unlike their classical counterpart. 
It allows a quantum error correcting code to correct errors even when they can not uniquely pinpoint the error.
The diagonal distance of a quantum code is an important parameter that characterizes whether the quantum code is degenerate or not.
If the code has a distance more than the diagonal distance, then it is degenerate; whereas if it is below the diagonal distance, then it is nondegenerate.
We show that most of the CWS codes without a cycle of length 4 attain the upper bound of diagonal distance $\delta+1$ where $\delta$ is the minimum vertex degree of the associated graph.
Addressing the question of degeneracy, we give necessary conditions on CWS codes to be degenerate.
We show that any degenerate CWS code with graph $G$ and classical code $C$, will either have a short cycle in graph $G$ or will be such that the classical code $C$ has one of the coordinates trivially zero for all codewords.}
\end{abstract}

\maketitle


\section{\label{sec:intro} Introduction}
\input{sec_1_intro}
 
\section{Notations and Preliminaries \label{sec:notation}}
\input{sec_2_notations}

\section{{Main Lemma} \label{sec:Main Lemma}}
\input{sec_3_main_lemma}

\section{{Degeneracy of CWS codes} \label{sec:degeneracy}}
\input{sec_4_degeneracy}

\section{\label{sec:further_app}{Further Application}}
\input{sec_5_further_applications}

\section{\label{sec:Conclusion} Conclusion}
\input{sec_6_Conclusion}
\nocite{*}

\bibliography{dia}

\end{document}

%% file: sec_1_intro.tex
In 1995, Shor~\cite{Shor_coding} showed that similar to classical computation, quantum computation can also be supported by an error correcting code. Soon after that, Calderbank and Shor~\cite{CS} and Steane~\cite{steane} came up with a standard procedure to combine two classical error correcting codes to obtain a quantum error correcting code.
 In his seminal work~\cite{Dan_thesis} Gottesman formalized the \emph{stabilizer} set up, giving a group theoretic framework to the study of quantum error correction.
 Although quite a few non-stabilizer codes were introduced later~\cite{nonstabilizer1,nonstabilizer2}, many of the good and popular error correcting codes, such as toric codes~\cite{toric_kitaev}, surface codes that are useful for large computation~\cite{Fowler} and recent quantum LDPC codes~\cite{RyanLD},  still come from the class of stabilizer codes.
 Right from the starting works of Calderbank, Shor and Steane, one of the most important directions for constructing new quantum codes has been fruitful use of classical error correction.
One of the most popular ways to obtain quantum codes is by using two classical codes such that one is contained in the dual of the other.
Codeword stabilized (CWS) codes~\cite{CWS} present another way to use classical codes.
Codeword stabilized codes are quantum codes made out of two classical objects: a graph, and a code.
The class of CWS codes contains all the stabilizer codes and also encompasses some non-stabilizer (or non-additive) quantum codes.
 
One of the main qualitative properties in which quantum codes differ from classical code is their `degeneracy'.
It allows a quantum code to correct more errors than it can uniquely identify. 
The Knill Laflamme theorem~\cite{Laflamme} that characterizes the distance of a quantum code intrinsically also considers this degenerate and nondegenerate nature of errors in the quantum setting.  
Consider a quantum code $\mathcal{M}$ with basis ${v_1,v_2,\ldots, v_{K}}$.
Then the Knill Laflamme condition for error correction involves looking at the error operator $E$ in the basis of $\mathcal{M}$, hence the entries $\langle v_i, E v_j \rangle$.
Roughly speaking, a code can uniquely identify all the errors $E$ where $\langle v_i, E v_j \rangle=0$.
Moreover, owing to quantum degeneracy, it can afford to have some overlap with the diagonal entries as long as it is constant over the diagonal for a fixed $E$ .
Let $\Delta^\prime(\mathcal{M})$ be the minimum (symplectic) weight of a Pauli error such that it has non-zero diagonal entries.
This quantity $\Delta^\prime$ is known as the diagonal distance of a quantum code.
The diagonal distance is a key combinatorial parameter of a quantum code.
All the codes with distance less than $\Delta^\prime$ are nondegenerate, whereas all those above it are degenerate.
Other than the theoretical interest of finding degenerate - nondegenerate boundaries of quantum error correction, diagonal distance also serves as an important tool for the construction of CWS codes by search~\cite{QuditGraph}. 
The method in~\cite{QuditGraph} fixes a graph and searches for good quantum codes for the fixed graph. 
Note that their method focuses only on searching for nondegenerate codes, and hence, it is favorable if we have a graph with a large diagonal distance since the distance of nondegenerate codes is upper bounded by their diagonal distance. Unfortunately, for a graph with minimum vertex degree $\delta$, its diagonal distance can at most be $\delta+1$.   
Nonetheless, one can ask a question, under what conditions can we go close to this upper bound of $\delta+1$.
Note that for a CWS code, its diagonal distance is wholly determined by its graph. 
Hence we use the terminology diagonal distance of a graph or code interchangeably for CWS codes.
In this work, we show that for a graph with no four cycles, we have a lower bound of $\delta$ on the diagonal distance.
Thus characterizing the diagonal distance of these graphs, we give the following theorem:
\begin{thmA} \label{thmA} Let $G$ be a graph with no four-cycle. 
If the minimum vertex degree of $G$ is $\delta$, then the diagonal distance is either $\delta$ or $\delta+1$. 
\end{thmA}
We would like to point out that we can further fine-tune \textbf{Theorem A} separating out cases with $\Delta^\prime=\delta$ and $\Delta^\prime=\delta+1$ (See Lemma~\ref{end_cor}). 
Thus we give a large class of graphs that attain the maximum diagonal distance bound.
 Furthermore, we also give necessary conditions for a quantum code to be degenerate.
 Informally, our theorem regarding degeneracy can be stated as follows: 
 \begin{thmB} For a degenerate CWS code $\mathcal{M}$ given by graph $G$ and classical code $C$, if $C$ uses all its components (that is, code is not trivially zero for all codewords in any of the components), then $G$ must contain a cycle of length 3 or 4. 
 \end{thmB}

Thus for a CWS code to be degenerate, either its classical code must be restricted to some hyperplane, or its graph must contain a short cycle.
We note that the actual statement we prove is much stronger than this.
We show that a graph contains either a four-cycle or a particular subgraph having high symmetry, which itself contains a three-cycle, hence the theorem. 
On the constructive side, as a direct application of the lemma~\ref{end_cor}, we can construct codes with distance $\Omega(\sqrt{n})$.
  
\paragraph{Why graphs without a four-cycle?}
One way to interpret our results is the necessity of short cycles in CWS graphs if one wants to exploit quantum degeneracy while using all the components of the code.
On the other hand, it also shows that if we focus only on graphs with no short cycles, then the nondegenerate limit is close to the maximum.  
Traditionally, graphs with a girth of more than four have enjoyed their fair share of success in classical coding theory, mainly due to LDPC codes~\cite{Gallager}. 
A standard decoder for LDPC codes known as the belief propagation decoder works only if the Tanner graph of the code does not contain short cycles.
Unfortunately, the natural extension with using two classical LDPC codes to construct a CSS code does not work.
The symplectic inner product condition ensures that such codes will always have a four-cycle. 
There has been some recent progress on this front, trying to implement belief propagation algorithms despite short cycles and degeneracy~\cite{DS_decode,kuo,BP2019}.  
Note that belief propagation algorithms come from a bigger class of message passing algorithms.
These algorithms work by sending `opinion' from a node to its neighbor about what label should the neighbor take. 
When there are short cycles, a wrong opinion is likely to be fed back, and hence the algorithm might not converge to a solution or may give an incorrect output.  
Note that, unlike CSS codes, graphs in CWS codes do not have any condition on short cycles. 
Moreover, our \textbf{Theorem A} shows that even in the nondegenerate region, where the qualitative nature of quantum errors resembles more with classical errors, we can still get codes with distance around the maximum limit of $\delta+1$. 
Thus, it will be interesting to see if the well-developed classical machinery of message passing algorithms can improve quantum error correction in the CWS framework.   
\subsection{Outline of Proof and Organization of the Paper}
In Section~\ref{sec:notation} we introduce basic background on CWS codes and some useful notations.
The key component in proof of \textbf{Theorem A} and \textbf{B} is the  main lemma (see Section~\ref{sec:Main Lemma}). 
It is well known that diagonal distance CWS codes can be understood by the kernel of the matrix $\left(\mathbb{I} \vert A_G\right)$ where $A_G$ is the adjacency matrix of the graph.
Hence linear dependence of columns of $\left(\mathbb{I} \vert A_G\right)$ completely characterize the diagonal distance.
First, we express the condition of a graph not having a four-cycle in terms of columns using \emph{Property A}(Section~\ref{sec:Main Lemma}).
We begin by showing that it suffices to show the main lemma for columns following a certain structure (Section~\ref{Sufficient}).  
The main lemma shows that every element in the kernel of $(\mathbb{I} \vert G)$ must be of at least one of the five forms given by the lemma. 
These conditions are labeled by $O,A.1,A.2,B$ and $C$.
The vector given by the condition $O$ is the trivial vector, whereas vectors given by $A.1$  and $A.2$ have symplectic weight at least $\delta+1$. 
Condition $B$ includes vectors with symplectic weight $\delta$ and occurs only in graphs with a very particular symmetry.
By noting that the adjacency matrix $A_G$ has no support along the diagonal, it is also fairly straightforward to show that vectors in condition $C$ have symplectic weight of at least $\delta+1$.
This gives us our \textbf{Theorem A}.
We then follow it up in Section~\ref{sec:degeneracy} with some implications of this lemma to the degeneracy of quantum CWS codes, including \textbf{Theorem B} and further application for code search in Section~\ref{sec:Conclusion}.

%% file: sec_2_notations.tex
In this section, we introduce basic notations and definitions from the theory of classical as well as quantum error correcting codes.
For more on the classical part, see~\cite{Huffaman} and for quantum codes see~\cite{Dan_thesis,Kitaev,Preskill_lecture}.
Let $\mathbb{F}_q$ be the finite field of size $q$.
We denote the set of matrices with $k$ rows and $n$ columns as $\mathcal{M}_{k,n}\left( R \right)$, where entries in the matrix are from the ring $R$.
To denote the set $\lbrace{0,1,\ldots, t-1 \rbrace}$, we use the notation $\left[ t \right]$.
For an $n$-dimensional vector $v$, we assume that it is represented as $v_1,\ldots,v_n$ so as to denote the projection on the $j$-th component by $v_j$.
Similarly, $A_{i,j}$ will denote the entry at row $i$ and column $j$ for matrix $A$.
We use $e_j$ to denote the $j$-th standard basis vector of an $n$-dimensional vector space; that is, $e_j$ is the vector with $1$ at j~th place and all the other entries being $0$.
\subsection{Coding Theory}\label{subsec:coding-theory}

\subsubsection{Classical Codes}
A linear error correcting code $\mathrm{C}$ is a $k$-dimensional linear subspace of an $n$-dimensional space.
Elements of the code $\mathrm{C}$ are called codewords.
Typically, one takes the larger space to be $\mathbb{F}_{q}^n$, and in that case, we say that $\mathbb{F}_{q}$ is the alphabet for underlying code.
Hamming weight and Hamming distance corresponding to a vector space $\mathbb{F}_q^n$ are respectively defined as follows :
\begin{gather*}
    wt_{H}(u)= \vert \lbrace u_i \neq 0\rbrace \vert\\
    d_H(u,v)= wt_H(u-v)
\end{gather*}
  The hamming weight of a vector is the number of coordinates in which the vector has a non-zero component, and the distance between two points is the number of components at which they differ.
Distance of a code is the minimum over the hamming distance between two distinct codewords.
  For a linear code, the distance of code is the same as the minimum weight of a non-zero codeword.
  A $k$-dimensional linear code with minimum distance $d$, sitting inside an $n$-dimensional subspace is denoted as $(n,k,d)$ code.
Since $\mathrm{C}$ is a $k$-dimensional subspace of $\mathbb{F}_q^n$, it can be thought of as the row space of a  matrix $\mathrm{G} \in \mathcal{M}_{k,n} \left(\mathbb{F}_q \right)$.
 Alternatively, $\mathrm{C}$ is the kernel of a \textit{parity check matrix} $\mathrm{H} \in \mathcal{M}_{n-k,n} \left( \mathbb{F}_{q}\right)$.
 Then the distance of a code is the minimum $t$ such that there exists a set of $t$ columns that are linearly dependent and any set of $t-1$ columns is linearly independent.
 A code with distance $d$ can correct upto $\lfloor \frac{d-1}{2}\rfloor$ errors.

A classical code is \emph{degenerate} if at least one of the components is trivially 0, that is, there exists an $i$ such that for all codewords $c\in C$, $c_i=0$.
Such a code $C$ carries no information in the component where it is degenerate.

\subsubsection{Quantum Codes}
A state of a quantum computer is often described by \textit{qubit(s)}.
A single qubit state is a unit vector in $\mathbb{C}^2$.
The space of multiple qubits is defined via the tensor product of corresponding individual qubit spaces, and thus, an n-qubit state is represented by a norm 1 vector in $\left({\mathbb{C}^2}\right)^ {\otimes n} \coloneqq \mathbb{C}^2 \otimes \cdots \otimes \mathbb{C}^2$; where the tensor product is taken over $n$ copies.
Thus, a state of an n-qubit system is a vector of norm 1 in $N= 2^n$ dimensional complex Hilbert space.

Let $v_0,v_1$ be an orthonormal basis for a single qubit.
Similarly, for $n-$~qubits, we use either $\vert v_j \rangle$ or $\vert v_{j_1,j_2,\ldots,j_n}\rangle$ where $j_1,j_2,\ldots,j_n$ is the $n-$bit representation of $j$.
For simplicity, when there is no confusion, we drop $v$ and represent $v_j$ simply as $j$.
For example, $\lbrace\vert 0 \rangle$, $\vert 1 \rangle \rbrace$ is a basis for single qubit whereas $\lbrace \vert 0\rangle,\vert 1\rangle,\vert 2\rangle,\vert 3\rangle\rbrace$ and  $\lbrace \vert 00\rangle,\vert 01\rangle,\vert 10\rangle,\vert 11\rangle \rbrace$ are two representations for the basis of 2-qubit systems.

Now we define two (families of) unitary operators on n-qubits. 
\begin{gather*}
Z(a): \mathbb{C}^N \rightarrow \mathbb{C}^N \hspace{0.8cm} Z(a)(u_j)= (-1)^{a.j}\ u_j  \hspace{0.6cm}  a\in \lbrace 0,1 \rbrace^n \\
    X(b): \mathbb{C}^N \rightarrow \mathbb{C}^N \hspace{0.8cm} X(b)\left(v_j \right)  \mapsto v_{j+b} \hspace{0.8cm} b\in \lbrace 0,1 \rbrace^n      
\end{gather*}
The \textit{n-qubit Pauli group} $\mathcal{P}_n$ is a subgroup of the Unitary group, generated by the above two families along with the scalars $\lbrace \pm 1, \pm i \rbrace$.
Formally,
 
 \[\mathcal{P}_n \coloneqq \lbrace i^t Z(a) X(b) : t\in \mathbb{Z}_4 , a,b \in\lbrace 0,1 \rbrace^n  \rbrace.\]
The \textit{support} of a Pauli operator $P$ is the set of coordinates where at least one of the vectors $a$ or $b$ is non-zero.
Support is exactly the set of qubits where $P$ acts non-trivially.
Any general operator on $\mathbb{C}^N$ can be expressed via the Pauli basis, $\mathcal{B}_p\coloneqq \lbrace P \in \mathcal{P}_n: t=0 \rbrace$.
Under this expansion, the base (or support) for any operator $L$ is the set of qubits in the Pauli expansion where the action is non-trivial.
The cardinality of support of E\@ is known as the \textit{size} (or weight) of E\@.

With this in hand, we are ready to describe the Knill-Laflamme condition for error correction.
A quantum code is a subspace $\mathcal{M} \subset  \mathcal{N}\coloneqq\left(\mathbb{C}^2 \right)^{\otimes n}$.
 Let $\mathcal{M}$ be a quantum code of dimension $K$ with basis $\lbrace v_i : i \in [K]\rbrace$.
 We  say that an operator $L$ satisfies the Knill-Laflamme condition if for all $i,j$ there exists a map $f$ which depends only on $L$ such that

\begin{equation}
 \langle v_i \vert L \vert v_j \rangle = \delta_{i,j} f(L) \label{KL}
\end{equation}
Then the Knill-Laflamme~\cite{Laflamme} theorem says that the code has distance $\delta$ if for all E\@ with size less than $\delta$ satisfy~\eqref{KL} and there exists an E\@ of size $\delta$ that violets~\eqref{KL}.
Similar to classical code, a quantum code can correct errors of size upto $\lfloor \frac{\delta-1}{2}\rfloor$.
And it can correct $\delta-1$ erasure errors.
The Knill-Laflamme condition can be looked at as a combination of two equations.
First, when $i \neq j$, the equation~\eqref{KL} says that $\langle v_i \vert L \vert v_j \rangle$,should be zero.
This is needed so that two different errors do not cause confusion. 
If $\langle v_i \vert L \vert v_j \rangle$ is non-zero then the original message $v_j$, after incurring error $L$ can potentially give $v_i$, causing confusion.
The second part of the condition, when $i=j$, indicates that all the diagonal elements must be the same (independent of $i$).
For a code to have distance $\delta$, both the conditions must be satisfied for all the operators of size $\delta-1$.

Categorically, based on these conditions  quantum codes can be divided in two sections.
Those that have $f(L)=0$ (for all $L$ of size upto $\delta-1$), known as \textit{non-degenerate} codes, and those for which there exists at least on $L=L_0$ of size smaller than $\delta$ with $f(L_0) \neq 0$,
known as \textit{degenerate} codes.
Note that classical and quantum notions of degeneracy are qualitatively different from each other. 
Classical degeneracy refers to one of the dimensions being unused, whereas quantum degeneracy tells that code is immune to a particular error.
In that case, although E\@ is a non-trivial error on $\mathbb{C}^N$, it acts trivially on the code $\mathcal{M}$.
This immunity to non-trivial error is not possible in the case of classical codes, as errors act via the addition of vectors.   
Hence, it becomes interesting to ask which quantum codes are degenerate and which are not.   
\subsection{Stabilizer Codes}\label{subsec:stabilizer-codes}
Consider the projection map from the Pauli group $\mathcal{P}_n$ to the binary vector space $\mathbb{F}_2^{2n}$,
\begin{gather*} \label{pro_map}
    \pi: \mathcal{P}_n \rightarrow \mathbb{F}_2^{2n}\\
    i^t~  Z(a)~  X(b)\mapsto \left(a,b\right)
\end{gather*}
The projection map $\pi$ is a homomorphism with kernel $\lbrace i^t \mathbb{I} \rbrace$, which is also the center of $\mathcal{P}_n$.

A subgroup $S \in \mathcal{P}_n$ is called a stabilizer subgroup if $S$ is abelian.
It is easy to see that $S$ is abelian if and only if $\pi(S)$ is self-orthogonal with respect to the symplectic inner product $\langle a,b \rangle_{s}:=a_2.b_1-b_2.a_1$ where $a=a_1 | a_2$ and $b=b_1 | b_2$.
For a vector $(a,b)$ its symplectic weight, $wt_{s}\left( a|b\right) = \vert \lbrace i:\ a_i=0\ or\ b_i=0 \rbrace \vert$.

 A subspace $\mathrm{Q} \subset \mathbb{C^{N}}$ is called a \textit{quantum stabilizer code} if there exists a stabilizer subgroup $\mathrm{S}$ which fixes every vector in $\mathrm{Q}$.
 We can alternatively define quantum code with respect to a stabilizer subgroup in the reverse direction; $\mathrm{Q}_\mathrm{S}\coloneqq \lbrace u \in \mathbb{C}^{N} \vert \ Pu=u \text{ for all } P\in \mathrm{S}\rbrace$.
Both of them are equivalent representations, and hence we can describe a stabilizer code by its corresponding stabilizer group.
 For more on stabilizer groups, we refer~\cite{Dan_thesis}.

Quantum stabilizer codes introduced by Gottesman~\cite{Dan_thesis} form an important subclass of quantum codes.
Their excellent algebraic structure allows one to use group theoretic techniques, and in many cases, they perform very close to non-stabilizer codes which lack this structure.

Decoding or detection of errors for stabilizers is done with syndrome measurement.
Consider a stabilizer with (minimal) generators $S_1,S_2,\ldots,S_k$.
Syndrome for an error $E$ is given by a $k$-dimensional binary vector, which indicates if $E$ commutes with each of the  $S_i$.
Thus an error $E$ is detected if it does not commute with at least one of the stabilizer generators.
Moreover, if $E$ itself is in the stabilizer group, then by definition, it acts trivially on the code, and hence code is immune to such errors. 
The only errors that can not be detected are those that commute with $S$ and are outside $S$.
Hence distance of $Q_S=\min \lbrace wt(E) : E \in \mathcal{N}(S) \setminus S \rbrace$ where $\mathcal{N}(S)$ denotes the normalizer of $S$ in the Pauli group.
\subsection{CWS codes}\label{subsec:cws-codes}

Codeword stabilized codes (\textbf{CWS}) is a more unrestricted framework for quantum error correcting codes that encompass the stabilizer framework.
CWS codes are described by two objects; a graph and a classical error correcting code.
Note that classical error codes in the CWS framework need not be linear.
 In fact, the original structure CWS codes were defined with a maximal stabilizer subgroup $S$ (of order $2^n$) and a set of $2^k$ Pauli elements referred to as \emph{word operators}.
Without loss of generality, the stabilizer group can be assumed to be generated by elements of the following form: $S_i= X(e_i) Z(r_i)$, that is, the stabilizer has only one $X$ operator and $r_i$ denotes where the action of the stabilizer has a $Z$-action. 
$r_i$'s form the rows of the adjacency matrix of the graph and word operators can be assumed to be of the form $Z({c_i})$ giving the classical error correcting component of the CWS framework.
This form with only one $X$ in each generator and word operators given by only $Z$ operators is the \emph{standard form} of CWS codes, for more see~\cite{CWS}.

CWS framework uses a graph to transfer $X$ errors into equivalent $Z$ errors, and hence all the errors can be treated as $Z$ errors. 
Roughly speaking, an $X$-error on a qubit is equivalent to $Z$ on all the neighbors given by the graph.
This transferring of errors to Z component only can be understood via the following map from the error set $\mathcal{E}$ to $n$-bits strings:
\[ Cl_{S}\left(E= Z(v) X(u)\right) =  v\ \oplus\ \bigoplus\limits_{i=1}^{n} u_i r_i \]

The $Cl_S$ map is a direct implementation of the above interpretation of the error transfer.
 All the $Z$ errors and all the rows corresponding to $X$ errors are added as binary vectors. 
\begin{theorem}\label{CWS condition} ~\cite[Theorem 3]{CWS} A CWS code in standard form with stabilizer $S$ and codeword operators $\lbrace Z(c) _{c \in C}\rbrace$ (where C is a classical error correcting code) detects errors from $\mathcal{E}$ if and only if $C$ detects errors from $Cl_S \left( \mathcal{E}\right)$  and in addition we have for each E\@,
\begin{gather*} Cl_S(E) \neq 0 \ \text{or} \\
\forall i,  Z(c_i) E= E Z(c_i)
\end{gather*}
\end{theorem}

The above theorem can also be understood as an analogue of Knill-Laflamme conditions for CWS codes.
Since for CWS codes, codewords are given by $c_i$ and stabilizer is a subgroup $S$ of order $2^n$, we can translate conditions as follows:
When $i \neq j$,   $\langle c_i E c_j = 0 \rangle$.
Hence, $Z(c_i) Z\left( Cl_S(E)\right) Z(c_j) \notin \pm S$.
The only element of $S$ that does not have an $X$ component is the identity $\mathbb{I}$, which translates the above condition to $Z(c_i) Z\left( Cl_S(E)\right) Z(c_j) \neq \mathbb{I}$   and can be rewritten as $c_i \oplus Cl_S(E) \neq c_j$.
This is exactly the same as the condition classical code $C$ detects $Cl_S(E)$.
On the other hand, for $i=j$,  an error can be detected if it does not commute with at least one of the stabilizers, $Cl_S(E) \neq 0$ or code is immune to it $Z(c_i) E=E Z(c_i)$.

A quantum code is degenerate or not depends on the last two conditions. 
Let $\Delta^\prime$ be the diagonal distance of a CWS code (or equivalently, a graph) defined as follows:
$\Delta^\prime = \min_{E \neq \mathbb{I}} \lbrace wt(E) : Cl_S(E) = 0 \rbrace$.
The diagonal distance plays a key part in the degenerate character of a quantum code.
If the distance of quantum code $d(\mathcal{M}) \leq \Delta^\prime$, then quantum code is non-degenerate otherwise it is degenerate.
Nondegenerate codes have a behavior very similar to classical linear codes, and hence they can benefit from rich machinery that is available for classical codes, whereas nondegeneracy is a feature unique to quantum codes which may give them benefit in terms of handling larger error sets.

Recall that map $\pi$ takes an error as input and outputs its corresponding $X$ and $Z$ components. 
For a CWS code defined by graph $G$ and classical code $C$, let $A_G$ denote the adjacency matrix of graph $G$.
It is easy to see that $Cl_S(E)=0$ if and only if $\pi(E) \in kernel \left(\mathbb{I} \vert A_G\right)$.
This follows directly from the definition of $Cl_S$ and the fact that $A_G$ is a symmetric matrix.
Hence, $\Delta^\prime(G)= \min_{x \neq 0} \lbrace wt_s(x)\ :\ x\in kernel(\mathbb{I} \vert A_G) \rbrace$.
In the following few sections, our focus will be on understanding the property $\Delta^\prime$ of graphs of a certain kind.

%% file: sec_3_main_lemma.tex
 Since the kernel of the matrix $\left( \mathbb{I}\vert A_G\right)$ characterizes its diagonal distance (and hence degeneracy), columns of this matrix play a crucial part.
Now we set up a few terminologies regarding them.

Let $A_j$ denote the $j$~th column of $A_G$.
We say that $i$ is in \textit{support} of $A_j$, denoted as \textit{supp$(A_j)$}, if $A_{ij}\neq 0$.
\begin{propertyA} Let $\mathcal{A}$ be a collection of distinct non-zero vectors (of length $n$).
We say that $\mathcal{A}$ has Property A if for any two distinct vectors $A_i,A_j$, $\vert supp(A_i) \cap supp(A_j) \vert \leq 1$. \end{propertyA}

This property is sometimes referred to as \textit{not more than one matching 1s} in the case of binary vectors.
It is easy to see that $G$ has no 4-cycle if and only if columns of $A_G$ have \textbf{Property A}. 
\begin{definition}[degree gap]
We say that a set $\mathcal{S}$ of columns has the degree gap $\delta (\geq 2)$ if for all $s \in \mathcal{S}$, $wt_{H}(s)=1$ or $wt_{H}(s) \geq \delta$ and there exists a column $s_0$ with $wt_{H}({s_{0}}) = \delta$.
\end{definition} 
The degree gap is a translation of the minimum vertex degree of a graph $G$ to the matrix $\left( \mathbb{I} \vert A_G\right)$.
If $G$ has a minimum degree $\delta \geq 2$ then it is the same as the degree gap of  $\left( \mathbb{I} \vert A_G\right)$.
Graph with isolated vertices or degree one vertices is not an interesting case as far as diagonal distance is concerned.
For such graphs, $\Delta^\prime$, which is the same as the minimum weight vector in the kernel, is very straightforward to deduce.
Note that for a graph with minimum degree $\delta$, its diagonal distance is upper bounded by $\delta+1$ by picking up minimum degree vertex and relevant columns from identity.

We say that $\Gamma= \lbrace \gamma_1, \ldots, \gamma_l\rbrace$, a set of vectors, \emph{sums to $0$}, if $\sum\limits_{i=1}^{i=l} \gamma_i$ is an all-zero vector.
To mean $\Gamma$ sums to $0$ and we write $\sum\Gamma=0$.
For any vector $S$, let $E(S)$ denote $\bigcup\limits_{j \in Supp(S)} e_j$.
That is, $E(S)$ consists of standard basis vectors corresponding to the support of $S$.

For proving \textbf{Theorem A}, we first prove the following lemma, from which the theorem follows directly : 

\begin{lemma}[Main Lemma] \label{main_lemma} Let $\mathcal{S}$ be a collection satisfying  \textbf{Property A} and degree gap $\delta$.
Then for any $\Gamma \subset \mathcal{S}$ such that $\sum \Gamma=0$ , at least of the following is true
\begin{itemize}
\item[O)] $\Gamma= \emptyset$. 
\item[A.1)] $ \Gamma$ contains at least $\delta+1$ columns of weight one.
\item[A.2)] $\Gamma$ contains at least $\delta+1$ columns of weight  greater than $1$. 
\item[B)] $\Gamma$ is a set of size $2\delta$ having $\delta$ columns  of weight $1$ and $\delta$ columns of weight greater than $1$. 

\item[C)] $\Gamma= S_\delta \cup E(S_\delta)$ for some vector $S_\delta$ of weight $\delta$.

\end{itemize}
\end{lemma}

Taking $\mathcal{S}$ to be columns of $\left[ \mathbb{I} \vert A_G\right]$ in Lemma~\ref{main_lemma}, any non-trivial element in the kernel will have symplectic weight at least $\delta$.

\subsection{Sufficient class of $\Gamma$ for Lemma~\ref{main_lemma}} \label{Sufficient}
In this section, we give a set of conditions for $\Gamma$. 
We show that it suffices to check Lemma~\ref{main_lemma} only for this class. 
Before that, we define the \textit{constraint Graph} corresponding to $\Gamma$.  
   
\begin{definition}[Constraint Graph] Given a set $\Gamma$ consider a bipartite graph $\mathcal{G}_\Gamma$  on $\vert\Gamma \vert +n$ vertices, $\mathcal{G}=\left( V_\Gamma \sqcup \mathrm{X}, \mathsf{E} \right)$ where $V_\Gamma$ is one partition of size $\vert \Gamma \vert$ having a vertex corresponding to each to a vector $v \in \Gamma$ and the other partition $\mathrm{X}= \lbrace 1,2,\ldots,n\rbrace$.
Edges define the support of vectors, that is, $\left( v_i,j\right) \in \mathsf{E}$ if $j$ is in the support of $v_i$.\end{definition}
 Let $\Gamma_1$ denote the vertices of $\Gamma$ that have degree one.
 These are precisely the vertices corresponding to standard basis vectors $e_j$ that are in $\Gamma$.
 Moreover, by construction, every vertex in $\Gamma_\delta:= V_\Gamma \setminus \Gamma_1$ has degree at least $\delta$.
 With the slight abuse of notation, we will denote $V_\Gamma$ by $\Gamma$ itself due to the natural one-to-one correspondence.
 For any vertex $u$ we will denote neighbors of $u$ by $\mathcal{N}(u)$ and $N(U)$ will denote $\bigcup_{u\in U} \mathcal{N}(u)$. \textbf{Property A} will translate to here as $\vert \mathcal{N}({v_i}) \cap \mathcal{N}({v_j}) \vert \leq 1$.
 We say that $v_i, v_j \in \Gamma$ intersect if this intersection is of size 1.
 If $\sum\Gamma=0$, then for all $x \in \mathrm{X}$, $\deg(x)$ is even.
 This added with \textbf{Property A}, will say that vertex $v$ of degree $d_v$ should intersect at least $d_v$ distinct vertices.
 
Let $\Gamma= \lbrace \alpha_1,\alpha_2, \ldots, \alpha_{\vert \Gamma_\delta \vert}, e_{i_{1}}, e_{{i}_2}, \cdots, e_{i_{\vert\Gamma_1\vert}} \rbrace$.
Without loss of generality, we can assume that $\deg(\alpha_i)\geq \deg(\alpha_j)$ for $i<j$.
 From here on, we will assume that $\Gamma$ satisfies the Property A and sums to $0$ unless stated otherwise. 

\begin{obs}\label{obs3.1} If there is a vertex $v= \alpha_i$ such that no  $e_j \in \Gamma_1$  intersects with $\alpha_i$ then $A.2)$ is satisfied.
This follows from the fact that $d_v \geq \delta$ and at least $d_v$ vertices should intersect with $v$.
So without loss of generality we assume that $e_{i_1}$ intersects $\alpha_1$.
\end{obs}

\begin{obs}\label{obs3.2} For $\Gamma$ of size greater than $2\delta$ either $A.1)$ or $A.2)$ are trivially satisfied and for $\Gamma$ of size $2\delta$ at least one of the $A.1)$, $A.2)$ or $B)$ are satisfied.  \end{obs}
\begin{obs}\label{obs3.3} Let $\Gamma$ be a collection such that $\vert \Gamma_\delta \vert= \delta$ then $\vert \Gamma_1 \vert \geq \delta$ and hence $\vert \Gamma \vert \geq 2 \delta$. 
\end{obs}
\begin{proof}
Let $\Gamma= \lbrace \alpha_1, \alpha_2, \ldots, \alpha_\delta, e_{i_{1}},\ldots, e_{i_{\vert \Gamma_1\vert}} \rbrace$.
Then we give an injective map $\gamma: \Gamma_{\delta} \hookrightarrow \Gamma_1$.
For any $\alpha_i$, let $\hat{Z}(\alpha_i):=supp(\alpha_i)\setminus \cup_{j\neq i} supp(\alpha_j)$, that is $Z(\alpha_i)$ denotes elements that are in support of only $\alpha_i$ and none of the other $\alpha_j$s. 
We can rewrite $\hat{Z}(i)$ as $supp(\alpha_i)- \cup_{j\geq 2} \left( supp(\alpha_i) \cap supp(\alpha_j)\right)$.
Note that as $\Gamma$ satisfies \textbf{Property A}, $\hat{Z}(\alpha_{i}) \neq \emptyset$.
Moreover, since $\Gamma$ sums to $0$, for each $z \in \hat{Z}(\alpha_i)$, $e_z\in \Gamma_1$.
By construction, each $z_j$ is in a unique $\hat{Z}(\alpha_i)$.
Now for every $i$, choose a  $z_i \in \hat{Z}(\alpha_i)$ now $\gamma: \alpha_i \mapsto z_i \in \hat{Z}(\alpha_i)$ is the required inclusion map.
\end{proof}
To summarize, from here on, without loss of generality, we can assume that

 $\mathcal{G}_\Gamma = \left( \Gamma \sqcup X, \mathsf{E} \right)$ where $X= \lbrace 1,2, \ldots, n \rbrace$ and
 $\Gamma= \lbrace \alpha_1,\alpha_2, \ldots, \alpha_{\vert \Gamma_\delta \vert}, e_{i_{1}}, e_{{i}_2}, \cdots, e_{i_{\vert\Gamma_1\vert}} \rbrace$ 
\begin{description}
\label{WLoG}
\item[1.] $\deg(\alpha_i) \geq \delta$, $\mathcal{N}({e_i})=i$ and for all $x\in X$, $\deg(x)$ is even.
\item[2.] For $i\neq j$, $\vert \mathcal{N}(\alpha_i) \cap \mathcal{N}(\alpha_j) \vert \leq 1$ (\textbf{Property A})
\item[3.] $\deg(a_i) \geq \deg(a_j)$ if $i<j$ \ldots (WLOG)
\item[4.] $i_1 \in \mathcal{N}(\alpha_1)$, each $v \in \Gamma$ intersects with at least $\deg(v)$ vertices of $\Gamma$. (Observation~\ref{obs3.1})
\item[5.] For all $\alpha_j$, there exists $e_{k} \in \Gamma$ such that $\alpha_j$ and $e_{k}$ intersect \ldots (Observation~\ref{obs3.1})
\item[6.] $\vert \Gamma \vert \leq 2\delta-1$(Observation~\ref{obs3.2}).
\end{description}
\subsection{Proof of Lemma~\ref{main_lemma}}
After reducing the class of $\Gamma$  to the above conditions, in this section, we show that Lemma~\ref{main_lemma} holds under these conditions.
We begin by separating out the $\Gamma$ of size $\delta+1$.

\begin{lemma} \label{no_low _weight} Let $\mathcal{S}$ be a collection satisfying \textbf{Property A} with degree gap $\delta$ and $\Gamma$ be a non-trivial subset of $\mathcal{S}$ of size at most $2 \delta-1$ with $\sum \Gamma=0$ then either
\begin{itemize}
\item[A)] $ \delta+2 \leq \vert \Gamma \vert \leq 2\delta-1$.
\item[C)] $\Gamma= S_\delta \cup E(S_\delta)$ for some vector $S_\delta$ of weight $\delta$.
\end{itemize}
\end{lemma} 
\begin{proof} Suppose $\vert \Gamma_1 \vert, \vert \Gamma_\delta \vert \leq  \vert \Gamma \vert \leq \delta+1$.
Then we will show that $C)$ holds.

Clearly for $x \in X$ to have an even degree, $\Gamma_1 \subsetneq \Gamma$.
So we get $\vert \Gamma_\delta \vert \geq 1$.
Suppose $\vert \Gamma_\delta \vert \geq 2$.
Let $a_1,a_2 \in \Gamma_\delta$ be two vertices of degree at least $\delta$.
 Define two sets $W$ and $W^\prime$ as follows:
\begin{gather*}
    W= \mathcal{N}(a_1) \cup \mathcal{N}(a_2)-\left( \mathcal{N}(a_1) \cap \mathcal{N}(a_2)\right) \\
    W^\prime= \bigcup_{\gamma \in \Gamma \setminus \lbrace a_1,a_2 \rbrace} \mathcal{N}(\gamma).
\end{gather*}
$W$ is the set of vertices that are neighbors of either $a_1$ or $a_2$ but not of both, and $W^\prime$ is the set of neighbors of the rest of the $\Gamma$.
The degree of any vertex in $W \setminus W^\prime$ is one.

Note that \[W \cap  W^\prime \subseteq \left(\mathcal{N}(a_1) \bigcup \mathcal{N}(a_2)\right)\ \bigcap \left( \bigcup\limits_{\gamma \in \Gamma \setminus \lbrace a_1,a_2 \rbrace} \mathcal{N}(\gamma) \right)\]
\[ = \bigcup\limits_{\gamma \in \Gamma \setminus \lbrace a_1,a_2 \rbrace}\left( \mathcal{N}(a_1)\bigcap  \mathcal{N}(\gamma) \right) \bigcup\limits_{\gamma \in \Gamma \setminus \lbrace a_1,a_2 \rbrace} \left( \mathcal{N}(a_2)\bigcap   \mathcal{N}(\gamma)\right)\]
Since each intersection is of size at most 1, we get $\vert W \vert \geq 2\delta-2$ and $\vert W \cap W^\prime  \vert \leq 2(\Gamma-2)$.
Thus, we get
$\vert W \setminus W^\prime \vert = \vert W \vert -\vert  W \cap  W^\prime \vert \geq 2\delta - 2(\Gamma-2) \geq 2 \delta - 2(\delta-1) \geq 2$.

But for $X$ to have an even degree, we need $W=W^\prime$ as vertices in $W\setminus W^\prime$ have degree one.
Contradiction.
Thus, the only remaining case is $\vert \Gamma_\delta\vert= 1$.
Let $\Gamma=v_o \cup \Gamma_1$ where $v_0$ has degree at least $\delta$.
 Again since all $x \in X$ have an even degree and $\vert \Gamma_1 \vert \leq \delta$, this vertex can not have a degree more than $\delta$. 
 This forces $\Gamma$ to be of the form $C)$.
\end{proof}
\begin{corollary} \label{size_delta1} There exists no $\Gamma$ of size less than $\delta+1$ such that it satisfies \textbf{Property A}, has degree gap $\delta$ and $\sum \Gamma=0$.
For size $\delta+1$, there are only two such $\Gamma$  either of the form $\Gamma= S_\delta \cup E(S_\delta)$ or having all the columns of weight $\delta+1$.
\end{corollary}
\begin{remark} This immediately gives $\Delta^\prime > \frac{\delta}{2} =\Theta\left( \delta\right)$ by taking $\mathcal{A}$ to be collection of columns of $\left[\mathbb{I} \vert A_G \right]$.   
\end{remark}

With Lemma~\ref{no_low _weight}, we eliminate the cases for $\Gamma$ with a size less than $\delta$.
Now we effectively are left with only $\Gamma$s in case $A)$ and thus, to prove Lemma~\ref{main_lemma}, it suffices to prove that condition $A)$ further breaks up into $A.1)$ and $A.2)$ giving Lemma~\ref{main_lemma}.
We prove this in multiple Lemmas.
We will prove it by induction on the gap between $\Gamma$ and $\delta$.
Recall that by construction, for every $\Gamma$, there exists a corresponding $\delta_\Gamma$ which is the minimum degree of vertex other than that of standard basis elements (degree one elements).
  Formally, let $X_\Gamma:= \vert \Gamma\vert- \delta_{\Gamma}-1$, we will show Lemma~\ref{main_lemma} by induction on $X_\Gamma$.
We start with the base case of $X_\Gamma=1$, that is, $\vert \Gamma \vert= \delta+2$.

\begin{lemma}\label{Base_case} For $\Gamma$ of size $\delta_\Gamma+2$ and following Conditions $1-5$ from Section~\ref{WLoG}  at least one $A.1)$ or $A.2)$ holds. \end{lemma}
\begin{proof}
We divide our proof into two cases, when all the $\alpha_i$s have the same degree and when they do not. 

\textbf{Case 1:} Let $\deg(\alpha_1) > \delta_\Gamma$

Consider $\tilde{\Gamma}:= \lbrace \alpha_1+e_{i_{1}} \rbrace \cup \Gamma \setminus \lbrace  \alpha_1, e_{i_{1}} \rbrace$.

That is, $\tilde{\Gamma}$ is obtained by replacing  $\alpha_1, e_{i_{1}} \rbrace$ with $\alpha_1+ e_{i_{1}}$ in $\Gamma$.
Clearly, $\Gamma$ sums to $0$ if and only if $\tilde{\Gamma}$ sums to $0$.
Also, since $\alpha_1$ intersects with $ e_{i_{1}}$ (condition 4), Property A also holds $\tilde{\Gamma}$.

Note that since $\deg(\alpha_1)>\delta_\Gamma$, the minimum degree in $\Tilde{\Gamma}$ is still $\delta_\Gamma$ but the size has decreased by $1$.
 Thus $\vert \tilde{\Gamma} \vert = \delta_{{\tilde{\Gamma}}} +1$.
  By Corollary~\ref{size_delta1}, $\tilde{\Gamma}$ has $\delta_{{\tilde{\Gamma}}}$ of weight 1, or it has $\delta_{{\tilde{\Gamma}}}+1$ columns of weight $\delta_{{\tilde{\Gamma}}}$.
   In the first case, $\Gamma$ contains  all the weight $1$ columns of $\tilde{\Gamma}$ and $e_{i_{1}}$ so it will satisfy $A.1)$ and in the second case it satisfies $A.2)$ as $\delta_{{\tilde{\Gamma}}}= \delta_{{{\Gamma}}}$.

\textbf{Case 2:} All $\alpha_i$ vertices have the same degree $\delta_\Gamma$.

Let $\vert \Gamma_{\delta_{\Gamma}} \vert=l$ and $\vert \Gamma_1\vert = \delta_\Gamma+2-l$.
If $l=1$ or $\delta_\Gamma+2-l=1$ we are already done. So suppose both of them are at least 2.

We further divide this in two sub-cases:

\textit{Case 2a)} There exists $j \neq 1 $ such that $\alpha_2$ and $e_{i_{j}}$ intersect.
Without loss of generality $j=2$.

Consider $\tilde{\Gamma}:= \lbrace \alpha_1+ e_{i_{1}},  \alpha_2+ e_{i_{2}}\rbrace \cup \Gamma \setminus \lbrace \alpha_1, e_{i_{1}}, \alpha_2, e_{i_{2}} \rbrace$

  With the same argument as above, $\sum\tilde{\Gamma}=0$ and satisfies $A$. From $\Gamma$ to $\tilde{\Gamma}$, size decreases by $2$ and minimum weight decreases by $1$; which gives $\vert \tilde{\Gamma} \vert=\delta_{\tilde{\Gamma}}+1$.

  By Corollary~\ref{size_delta1}, $\tilde{\Gamma}$ contains exactly one vertex of degree $\delta_{\tilde{\Gamma}}$ or all vertices are of degree $\delta_{\tilde{\Gamma}}$.
  Now $\tilde{\Gamma}$ has  at least two vertices of degree $\delta_{{\Gamma}}-1$, namely, $\alpha_1+ e_{i_{1}}$ and $ \alpha_2+ e_{i_{2}}$ and at least $l-2$ vertices of degree $\delta_{{\Gamma}}$ namely $\alpha_{i\geq3}$ and at least $ \delta_{\Gamma}-l$ vertices of degree $1$ namely $e_{i_{k \geq 3}}$.
   First case forces, $\deg(\alpha_1)=2$ and thus $\vert \Gamma \vert=4$ and .
   The second case also forces $\delta_{\Gamma}=l$ as it can not have a vertex of degree 1 and $l=2$ as it can not have a vertex of degree $\delta_\Gamma$ and hence in this case also $\vert \Gamma\vert=4$.
  
\textit{Case 2b)} For any $j \neq 1$, $\alpha_2$ does not intersect with $e_{i_{j}}$.
By condition 5, $1 \in \mathcal{N}(\alpha_2)$ and hence $\mathcal{N}(\alpha_1) \cap \mathcal{N}(\alpha_2)= \lbrace 1 \rbrace $.
Now consider the sets $R=\lbrace \alpha_{j\geq 3} \rbrace$ and $T= \mathcal{N}(\alpha_2)\setminus\lbrace 1 \rbrace$. $T$ has size $\delta_\Gamma -1$ and each $t\in T$ has degree at least one.
 It is easy to see that at least $\delta_{\Gamma}-1$ distinct vertices from $R$ must intersect with $\alpha_2$ and thus $\vert \Gamma_{\delta_{\Gamma}} \vert = \vert R \vert +2 \geq\delta_{\Gamma}-1+2=  \delta_{\Gamma}+1$.

\end{proof}
\begin{corollary}\label{coro_delta+1} For all $\Gamma$ of size $\delta+2$ that sum to $0$ and satisfy property A, either $\vert \Gamma_1 \vert \geq \delta+1$ or $\vert \Gamma_\delta \vert \geq \delta+1$. 
\end{corollary}
\begin{proof}
The proof follows directly from the reasoning of Observations~$\left(\ref{obs3.1},\ref{obs3.2}\right)$ combining with Lemma~\ref{Base_case}.
\end{proof}
Now we are ready for the main lemma.
The proof will go along the lines similar to Lemma~\ref{Base_case} and inducting over the difference between the size and the gap of $\Gamma$.

We recall the structure and conditions on $\Gamma$ under which it is sufficient to prove.

$\mathcal{G}_\Gamma = \left( \Gamma \sqcup X, \mathsf{E} \right)$ where $X= \lbrace 1,2, \ldots, n \rbrace$ and

 $\Gamma= \lbrace \alpha_1,\alpha_2, \ldots, \alpha_{\vert \Gamma_{\delta_{\Gamma}} \vert}, e_{i_{1}}, e_{{i}_2}, \cdots, e_{i_{\vert\Gamma_1\vert}} \rbrace$ 
\begin{description}
\item[1.] $\deg(\alpha_i) \geq \delta_\Gamma$, $\mathcal{N}({e_i})=i$ and  $\deg(x \in X)$ is even
\item[2.] For $i\neq j$, $\vert \mathcal{N}(\alpha_i) \cap \mathcal{N}(\alpha_j) \vert \leq 1$ \ldots  (\textbf{Property A})
\item[3.] $\deg(a_i) \geq \deg(a_j)$ if $i<j$ \ldots (WLOG)

\item[4.] $i_1 \in \mathcal{N}(\alpha_1)$, each $v \in \Gamma$ intersects with at least $\deg(v)$ vertices of $\Gamma$\ldots (Observation~\ref{obs3.1})

\item[5.] For all $\alpha_j$, there exists $e_{{k}} \in \Gamma$ such that $\alpha_j$ and $e_{{k}}$ intersect
\item[6.] $\delta_\Gamma +2 \leq \vert \Gamma \vert \leq 2\delta_{\Gamma}-1$ \ldots (Observation~\ref{obs3.2}, Lemma~\ref{no_low _weight})

\end{description}

\begin{lemma} \label{end} For the $\Gamma$ of size at least $\delta_{\Gamma}+2$, at most of size $2 \delta_\Gamma-1$ such that sums to $0$ and satisfies $Property\ A$ either $\vert \Gamma_1\vert \geq \delta_\Gamma+1$ or $\vert \Gamma_{\delta_{\Gamma}}\vert \geq \delta_\Gamma+1$.
To put it alternatively, for $\Gamma$ satisfying conditions $1,2,6$ above, $\max{\lbrace \vert \Gamma_1\vert ,\vert \Gamma_{\delta_{\Gamma}} \vert\rbrace \geq \delta_\Gamma+ 1}$
\end{lemma}
\begin{proof} 
We prove this by induction on $X:= \vert \Gamma \vert- \delta_\Gamma -1$.

Case $X=1$ is done by Corollary~\ref{coro_delta+1}. 

\textbf{ Induction hypothesis}~(IH): for all $\tilde{\Gamma}$ satisfying $1,2,6$ and   $X(\tilde{\Gamma}) = r-1$ ,
 either $\vert \tilde{\Gamma}_1\vert \geq \delta_{\tilde{\Gamma}}+1$ or $\vert \Gamma_{\delta_{\tilde{\Gamma}}}\vert \geq \delta_{\tilde{\Gamma}}+1$.
 
We will show that for $\Gamma$ with $X_\Gamma=r >1$, the lemma holds.
Note that we can assume without loss of generality that $3,4,5$ also holds for $\Gamma$.

\textbf{Case 1: } $\deg(\alpha_i) > \delta_\Gamma$.

Consider $\tilde{\Gamma}=  \lbrace \alpha_1+e_{i_{1}} \rbrace \cup \Gamma \setminus \lbrace \alpha_1, e_{i_{1}} \rbrace$.

Clearly $\tilde{\Gamma}$ satisfies $1-2$ with $\delta_{\tilde{\Gamma}}= \delta_{\Gamma}$ and $\vert \tilde{\Gamma}\vert = \vert \Gamma\vert-1$.
Thus $X(\tilde{\Gamma})=r-1$.
Moreover $\vert \Gamma_1\vert \geq \vert \tilde{\Gamma}_1\vert+1$ and $\vert \Gamma_{\delta_{\Gamma}}\vert \geq \vert \tilde{\Gamma}_{\delta_{\tilde{\Gamma}}} \vert$.
Now if $\tilde{\Gamma}$ satisfies $6$ then by \textbf{IH} we are done.
If it does not, then $\vert \tilde{\Gamma}\vert= 2 \delta_{\tilde{\Gamma}}$ which gives $\max{\lbrace\vert \tilde{\Gamma}_1\vert+1, \vert  \tilde{\Gamma}_{\delta_{\tilde{\Gamma}}}\vert\rbrace} \geq \delta_{\tilde{\Gamma}}+1$.
Hence, $\max \lbrace \vert \Gamma_1 \vert, \vert \Gamma_{\delta_{\Gamma}}\vert \rbrace \geq \delta_{\Gamma}+1$.

\textbf{Case 2:} All the vertices have the same degree $\deg(\alpha_i)=\delta_\Gamma$.
We divide this further into two parts.

\noindent \underline{{Case 2.a)}} There exists $j\neq 1$, $\alpha_j$ such that $\alpha_2$ intersects with $e_{i_{j}}$.
 Without loss of generality $j=2$.

Consider $\tilde{\Gamma}:= \lbrace \alpha_1+ e_{i_{1}},  \alpha_2+ e_{i_{2}}\rbrace \cup \Gamma \setminus \lbrace \alpha_1, e_{i_{1}}, \alpha_2, e_{i_{2}} \rbrace$

Again, $\tilde{\Gamma}$ satisfies $1-2$ with $X(\tilde{\Gamma})= \vert \tilde{\Gamma}\vert- \delta_{\tilde{\Gamma}}-1= \left(\vert \Gamma \vert-2\right)- \left(\delta_{\Gamma}-1\right)-1= r-1$.
 Note that in this case, $\vert \Gamma_1\vert \geq \vert \tilde{\Gamma}_1\vert+2$ and $\vert \Gamma_{\delta_{\Gamma}}\vert \geq \vert \tilde{\Gamma}_{\delta_{\tilde{\Gamma}}} \vert$.
  So if $\vert \tilde{\Gamma}\vert$ satisfies $6$ 
  then by \textbf{IH}, $\max \lbrace \vert \tilde{\Gamma}_{\delta_{\tilde{\Gamma}}} \vert, \vert \tilde{\Gamma}_1 \vert\rbrace \geq \delta_{\tilde{\Gamma}}+1$.
  Hence $\max \lbrace \vert \Gamma_1\vert -2 ,\vert \Gamma_{\delta_{\Gamma}} \vert\rbrace \geq \delta_\Gamma$.
  Note that by Observation~\ref{obs3.3}, $\vert \Gamma_{\delta_{\Gamma}} \vert \neq \delta_{\Gamma_{\delta}}$.
  Thus if $\vert \Gamma_1\vert -2 \leq \vert \Gamma_{\delta_{\Gamma}} \vert$ then $\max$ is $\vert \Gamma_{\delta_{\Gamma}} \vert \neq \delta_{\Gamma_{\delta}}$.
  Otherwise,
    $\vert \Gamma_1 \vert -2 \geq \Gamma_{\delta_{\Gamma}}$ and hence $\max\lbrace \vert \Gamma_1\vert ,\vert \Gamma_{\delta_{\Gamma}} \vert\rbrace \geq \delta_\Gamma+2.$
 
\noindent \underline{{Case~2.b)}} Exactly same as Case 2b) of Lemma~\ref{Base_case}. \end{proof}
 
Recall that it suffices to check Lemma~\ref{main_lemma} for $\Gamma$ satisfying conditions 1-6, and it follows for the rest of an arbitrary $\Gamma$.
Thus, for any $\Gamma \subseteq \mathcal{S}$, we have shown that~$\ref{main_lemma}$ holds and that completes our proof of Lemma~\ref{main_lemma}.
 \paragraph{finishing up the proof} \label{finishing up}
Apply Lemma~\ref{main_lemma} on $\mathcal{S}$ to be columns of the matrix $\left[\mathbb{I} \vert A_G\right]$.
Let $\mu\left({\Gamma}\right)$ be the characteristic vector (over columns) for $\Gamma$.
That is, $\mu(\Gamma)(i)=1$ iff $i$th column of $\left[\mathbb{I} \vert A_G \right]$ is in $\Gamma$.
Clearly, $\sum\Gamma=0$ iff $\mu(\Gamma) \in \ker\left(\mathbb{I} \vert A_G\right)$.
Recall that diagonal distance is the minimum symplectic weight for a non-trivial vector in the kernel.
By Lemma~\ref{main_lemma}, $\Gamma$ must satisfy at least one of the conditions $O, A.1, A.2, B$ or $C$.
Condition $O$ indicates $\mu\left(\Gamma\right)$ is the trivial vector.
By definition, for $\Gamma$ satisfying condition $A.1$ or $A.2$, symplectic weight is at least $\delta+1$. It is not hard to see that the condition $C$ also has a symplectic weight of $\delta+1$.
 The adjacency matrix $A_G$ has all diagonal entries $0$, hence if $S_\delta$ is the $i$th column of $A_G$ then $e_i \notin E(S_\delta)$ which gives the symplectic weight of $\delta+1$.
 Condition $B$ is the only case that can have symplectic weight $\delta$.
  And thus, in either of the cases, diagonal distance is lower bounded by $\delta$.
   
We note that this can be improved further since case $B$ can be further broken up into cases having diagonal distance $\delta$ and those having $\delta+1$.
Accounting for the locations of columns, $G$ has diagonal distance $\delta$ only if $\left[ \mathbb{I} \vert A_G\right]$ contains $\delta$ columns of weight $\delta$ such that every pair has common support and no two differ. 
These are precisely the cases of size $2\delta$ we get, after replacing inequalities in conditions of section~\ref{Sufficient} by equalities.
These cases have high symmetry and regularity, and they form a tiny fraction of all the cases, and hence, most of the graphs with no 4-cycle, even from case $B$, saturate the diagonal distance bound.
The following lemma states those rare cases of high symmetry where diagonal distance is $\delta$. 
\begin{lemma} \label{end_cor}Let $G= \left( \mathsf{V}, \mathsf{E}\right)$ be a $4$-cycle free graph with minimum degree $\delta$.
Then $\Delta^\prime(G)= \delta$ if and only if it contains a set of vertices $V^\prime \subseteq V$ of size $\delta$ satisfying the following:
\begin{enumerate}[1.]
\item All the vertices in $V^\prime$ are vertices of minimum degree. That is, for all $v \in V^\prime$, the  degree of $v$ (computed in $G$) is $\delta$.
Additionally, $\delta$ must be even.
\item for every distinct pair $u,v \in V^\prime$, there exists a path of distance two (in G) between them.
Moreover all of these paths are disjoint, that is they do not share an edge.
\item for each $v \in V^\prime$, the set $V^\prime$ contains exactly one neighbor of $v$.
\end{enumerate}
Otherwise, for graph $G$ with no 4-cycle $\Delta^\prime(G)=\delta+1$.
\end{lemma}
\begin{proof} Proof follows from the fact that the only condition can give rise to a vector in the kernel with symplectic weight $\delta$ is condition $B$.
Such a vector must have $\delta$ columns of weight $\delta$ and $\delta$ columns of weight $1$.
Suppose columns of weight $\delta$ are $\alpha_1,\alpha_2,\ldots,\alpha_\delta$ and weight 1 columns are $e_{\beta_1},e_{\beta_2},\ldots,e_{\beta_\delta}$.
Let $\mathcal{A}=\sum_{i} \alpha_i$ and $\mathcal{B}=\sum_i \beta_i$.
Now with the kernel, we get, $\mathcal{A}= \mathcal{B}$.
Hence $\vert supp\left(\mathcal{A}\right) \vert= \vert supp\left(\mathcal{B}\right) \vert = \delta$.
Now note that  because of \textbf{Property A}, for each $\alpha_i$, we get $\vert supp(\alpha_i) \cap supp(\mathcal{A}) \vert \geq 1$.
Hence, at least one non-zero from each column survives. 
Note that more than one entry can not survive since the support of sum $\mathcal{A}$ is $\delta$.
Thus exactly one entry from each column survives, and this can only happen if each pair of columns intersects and every column is of weight $\delta$.
Let $\vert supp(\alpha_i) \cap supp(\alpha_j) \vert =h_{ij}$ for $i \neq j$.
Then vertex $h_{ij}$ is the neighbor of both $i$ and $j$, giving a path of length 2.
Disjointness of the path follows from the same argument showing each pair must intersect in distinct and unique support.
This shows $2.$
Now we can assume this $\Gamma$ to be in the following form:

$\Gamma = \alpha_1,\alpha_2,\ldots,\alpha_\delta ,e_{\beta_1},e_{\beta_2},\ldots,e_{\beta_\delta}$ such that  
 $\alpha_i$ corresponds to adjacency list of $v_{\beta_i}$ and for each $\alpha_i$ there exists unique $\beta_j$ such that $\beta_j \in supp(\alpha_i)$. 
Thus, $A_G[\alpha_i][\beta_j]=1$ and hence, $v_{\alpha_i}$ and $v_{\beta_j}$ are neighbors.
This gives a unique neighbor for each $\alpha_i$ proving 3.
Since there is exactly one neighbor for each $\alpha_i$ it follows that they occur in pairs and hence $\delta$ must be even.
\end{proof}

%% file: sec_4_degeneracy.tex
\begin{lemma}\label{degeneracy_conditions} Let $\mathcal{M} = \left(G,C \right)$ be a CWS code that is degenerate then at least one of the following is true:
\begin{itemize}
\item[1.] $G$ contains a cycle of length 3 or 4.
\item[2.] $C$ is a classical degenerate code.
\end{itemize}
So, a degenerate code with a non-degenerate classical component must have a short cycle in the graph.
\end{lemma}
\begin{proof}Let $(G,C)$ be the CWS codes with $S= \lbrace S_1, S_2,\ldots, S_n\rbrace$ as the generators of the stabilizer group in standard form.
We will show that if $G$ does not have a cycle of length four and $C$ is nondegenerate, then $G$ will have a cycle of length 3.
It is very well known that if classical code uses the $i$-th bit, in other words, there exists a codeword $c\in C$ such that $c_i\neq 0$, then the distance of CWS code is upper bounded by $wt(S_i)$.
Now Suppose $C$ is nondegenerate. 
This gives us $d\left( \left(G, C \right)\right) \leq \delta+1$.
Note that for $(G,C)$ to be degenerate we must have $\Delta^\prime\left(G\right) < d\left( \left(G, C \right)\right) \leq \delta+1$.
Thus for a degenerate quantum code $(G,C)$ with nondegenerate $C$ we will have $\Delta^\prime (G) \leq \delta$. 
Now from Lemma~\ref{end_cor}, if graph $G$ is four-cycle free and $\Delta^\prime \leq \delta$, then every $u \in V$ must have a neighbor in $V$ (condition 3). 
Furthermore, there must be a path of length $2$ between $u$ and $v$. 
This gives a triangle between $u,v$ and some vertex $w$ such that $u-w-v$ is a path in $G$.
\end{proof}
This completes our proof of the necessary condition for quantum degeneracy. 
Note that the actual statement we can prove is stronger than the mere existence of short cycles.
Using~\ref{end_cor}, we get that if $G$ does not contain a cycle of length 4 and $C$ is nondegenerate, then $G$ must contain a particular subgraph of size $\delta$.
This subgraph further contains a triangle, giving us lemma~\ref{degeneracy_conditions}.
Depending on the use case, we might consider other properties of the subgraph; for example, it not only contains one triangle but contains at least $\frac{\delta}{2}$ disjoint triangles (sharing no common vertex), and it also contains at least $\delta$ vertices of degree $\delta$.

We give a finer version of the above lemma which captures degeneracy in a slightly better manner.
\begin{corollary} A CWS code without short cycles (of length $\leq 4$), is degenerate only if classical code is degenerate in all the minimum degree vertices.
\end{corollary}
\begin{proof} Note that if $G$ does not contain cycles of short length then by lemma~\ref{end_cor} and argument similar to lemma~\ref{degeneracy_conditions}, $\Delta^\prime(G)=\delta+1$.
Suppose $(G,C)$ is degenerate then it must detect errors of weight $\delta+1$. 
Also for each $i \in v_{min}$, $Cl_S(S_i)= 0$ as $S$ is abelian.
Thus $S_i$ from vertices corresponding to the minimum degree are errors of weight $\delta+1$ with trivial $Cl_S$.
Again from Theorem~\ref{CWS condition}, they must commute with each of the $Z(c)$.
Suppose $C$ was not degenerate in one of the $v_{min}$ vertices labeled as $v_t$, then choose $c$ such that $c_t \neq 0$.
Recall that $S_t$ is in the standard form, and hence,  has $X$ term only at $v_t$ and hence $Z(c_t) S_{t} \neq S_{t} Z(c_t)$ giving the contradiction.
Hence $C$ must be degenerate in all minimum degree vertices. 
\end{proof}
Lemma~\ref{degeneracy_conditions} gives a set of necessary conditions for CWS code to be degenerate.
 It will be interesting to see to what extent the converse holds.
In the reverse direction, we can see that for every classical degenerate code $C$, there exists a graph $G$ such that $(G, C)$ is quantum degenerate. 
Let $C$ be degenerate in $i$-th component, construct a graph on $n$ vertices without any short cycle and small enough minimum degree $\delta$ (label this vertex as $v_i$) and ensure that all the other vertices have a higher degree higher than $\delta+2$.
You can apply Lemma~\ref{main_lemma} repeatedly once on $\Gamma$ including $v_i$ and once excluding $v_i$ to see that such a code will have distance $\delta+2$.
One can easily construct such a graph for $\delta= \Theta\left( \sqrt{n}\right)$ which will give degenerate code with $\Omega(\sqrt{n})$ distance.

%% file: sec_5_further_applications.tex
Apart from understanding quantum degeneracy, our main lemma can be used as a tool to construct quantum codes.
\paragraph{Searching Codes}
For a given graph $G$, one can potentially try to iteratively select $c_0,c_1,c_2$, and so on, maintaining the property of distance between them, trying to fit as big code as we can.
The spirit of the problem is very similar to that of a packing problem, and the process need not give an optimum $K$ for the obvious reason that at some point, we could have chosen a codeword that does not allow us to reach maximum dimension.
In~\cite{QuditGraph} this problem is marginally circumvented by focusing on nondegenerate codes.
First, creating a lookup table for all the Pauli distances, $\Delta^\prime(G)$ is found.
Construct a graph $\overline{G}$ whose vertices are labeled by graph basis states. 
Fix a parameter $d \leq \Delta^\prime$ and connect two basis states if the distance between them is at least $d$.
A valid quantum codeword (word operators in CWS language) then refers to a clique in $\overline{G}$. 
This then yields a quantum non-degenerate code with distance $d \leq \Delta^\prime$.
Thus for this method, it becomes vital to choose $G$ that has a large $\Delta^\prime$.
Our lemma~\ref{end_cor} is a useful step in this direction.
It shows that $\Delta^\prime$ can be made very close to its upper limit of $\delta+1$.
In fact, for most graphs without a four the cycle, it attains $\delta+1$.
Note that although the method in~\cite{QuditGraph} is better than brute search, it still involves constructing a lookup table and finding a maximum clique that is computationally hard on a general instance.
They raise the question, `can we exploit some properties or structures which can speed things up?' 
Lemma~\ref{end_cor} gives us a way to do so for graphs with no short cycle, and obtain a family quantum non-degenerate codes with $\Theta(\sqrt{n})$.
Since diagonal distance for such graphs reaches $\delta$, a code will detect an error $E$ with a weight less than $\delta$ if the classical code has a distance more than $Cl_S(E)$.
Consider $Cl_S(E)$ for any $E$ of symplectic weight at most $\delta$.
Let the maximum degree of a vertex in $G$ be $\delta_{max}$. 
Then for $Cl_S(E)$ will have hamming weight at most $(\delta+1) \delta_{max}$.
So a classical code that detects $(\delta+1)\delta_{max}$ errors will give a CWS code with distance $\Omega\left( \delta\right)$.
For example, one can start with a classical liner code family with linear distance. After picking up right constants, and choosing a \emph{close to regular} graph $\frac{\sqrt{n}}{c_1} \leq \delta \leq \delta_{max} \leq \frac{\sqrt{n}}{a_2}$, we can get a family of quantum codes with distance $\delta= \Theta\left( \sqrt{n}\right)$, without going through the hassle of the lookup table and clique search.

%% file: sec_6_Conclusion.tex
We give a complete characterization of the diagonal distance of CWS code $\mathcal{M}= (G, C)$ when the girth of $G$ is more than four.
All of these graphs have high diagonal distance, approaching the upper bound of $\delta+1$ where $\delta$ is the minimum vertex degree of $G$.
Most of these graphs have diagonal distance $\delta+1$, whereas only graphs having a particular strong symmetry have diagonal distance $\delta$. 
We also show that for a CWS code $(G, C)$ to be quantum degenerate, either $C$ must be classically degenerate or $G$ must have a short cycle.
Moreover, we give finer versions of the result by showing where $C$ must be degenerate for $G$ without a short cycle.  
It would be interesting to see to what extent these conditions are tight. 
Can we give a converse by providing necessary conditions for $(G, C)$ to be degenerate? 
Given a graph $(G, C)$, how efficiently can we decide whether it is degenerate or not?

%% file: Diagonal.bbl
\begin{thebibliography}{24}%
\makeatletter
\providecommand \@ifxundefined [1]{%
 \@ifx{#1\undefined}
}%
\providecommand \@ifnum [1]{%
 \ifnum #1\expandafter \@firstoftwo
 \else \expandafter \@secondoftwo
 \fi
}%
\providecommand \@ifx [1]{%
 \ifx #1\expandafter \@firstoftwo
 \else \expandafter \@secondoftwo
 \fi
}%
\providecommand \natexlab [1]{#1}%
\providecommand \enquote  [1]{``#1''}%
\providecommand \bibnamefont  [1]{#1}%
\providecommand \bibfnamefont [1]{#1}%
\providecommand \citenamefont [1]{#1}%
\providecommand \href@noop [0]{\@secondoftwo}%
\providecommand \href [0]{\begingroup \@sanitize@url \@href}%
\providecommand \@href[1]{\@@startlink{#1}\@@href}%
\providecommand \@@href[1]{\endgroup#1\@@endlink}%
\providecommand \@sanitize@url [0]{\catcode `\\12\catcode `\$12\catcode
  `\&12\catcode `\#12\catcode `\^12\catcode `\_12\catcode `\%12\relax}%
\providecommand \@@startlink[1]{}%
\providecommand \@@endlink[0]{}%
\providecommand \url  [0]{\begingroup\@sanitize@url \@url }%
\providecommand \@url [1]{\endgroup\@href {#1}{\urlprefix }}%
\providecommand \urlprefix  [0]{URL }%
\providecommand \Eprint [0]{\href }%
\providecommand \doibase [0]{https://doi.org/}%
\providecommand \selectlanguage [0]{\@gobble}%
\providecommand \bibinfo  [0]{\@secondoftwo}%
\providecommand \bibfield  [0]{\@secondoftwo}%
\providecommand \translation [1]{[#1]}%
\providecommand \BibitemOpen [0]{}%
\providecommand \bibitemStop [0]{}%
\providecommand \bibitemNoStop [0]{.\EOS\space}%
\providecommand \EOS [0]{\spacefactor3000\relax}%
\providecommand \BibitemShut  [1]{\csname bibitem#1\endcsname}%
\let\auto@bib@innerbib\@empty
\bibitem [{\citenamefont {Shor}(1995)}]{Shor_coding}%
  \BibitemOpen
  \bibfield  {author} {\bibinfo {author} {\bibfnamefont {P.~W.}\ \bibnamefont
  {Shor}},\ }\href@noop {} {\bibfield  {journal} {\bibinfo  {journal} {Physical
  review A}\ }\textbf {\bibinfo {volume} {52}},\ \bibinfo {pages} {R2493}
  (\bibinfo {year} {1995})}\BibitemShut {NoStop}%
\bibitem [{\citenamefont {Calderbank}\ and\ \citenamefont {Shor}(1996)}]{CS}%
  \BibitemOpen
  \bibfield  {author} {\bibinfo {author} {\bibfnamefont {A.~R.}\ \bibnamefont
  {Calderbank}}\ and\ \bibinfo {author} {\bibfnamefont {P.~W.}\ \bibnamefont
  {Shor}},\ }\href@noop {} {\bibfield  {journal} {\bibinfo  {journal} {Physical
  Review A}\ }\textbf {\bibinfo {volume} {54}},\ \bibinfo {pages} {1098}
  (\bibinfo {year} {1996})}\BibitemShut {NoStop}%
\bibitem [{\citenamefont {Steane}(1996)}]{steane}%
  \BibitemOpen
  \bibfield  {author} {\bibinfo {author} {\bibfnamefont {A.~M.}\ \bibnamefont
  {Steane}},\ }\href@noop {} {\bibfield  {journal} {\bibinfo  {journal}
  {Physical Review Letters}\ }\textbf {\bibinfo {volume} {77}},\ \bibinfo
  {pages} {793} (\bibinfo {year} {1996})}\BibitemShut {NoStop}%
\bibitem [{\citenamefont {Gottesman}(1997)}]{Dan_thesis}%
  \BibitemOpen
  \bibfield  {author} {\bibinfo {author} {\bibfnamefont {D.}~\bibnamefont
  {Gottesman}},\ }\href@noop {} {\bibfield  {journal} {\bibinfo  {journal}
  {arXiv preprint quant-ph/9705052}\ } (\bibinfo {year} {1997})}\BibitemShut
  {NoStop}%
\bibitem [{\citenamefont {Nirkhe}\ \emph {et~al.}(2018)\citenamefont {Nirkhe},
  \citenamefont {Vazirani},\ and\ \citenamefont {Yuen}}]{nonstabilizer1}%
  \BibitemOpen
  \bibfield  {author} {\bibinfo {author} {\bibfnamefont {C.}~\bibnamefont
  {Nirkhe}}, \bibinfo {author} {\bibfnamefont {U.}~\bibnamefont {Vazirani}},\
  and\ \bibinfo {author} {\bibfnamefont {H.}~\bibnamefont {Yuen}},\ }in\ \href
  {https://doi.org/10.4230/LIPIcs.ICALP.2018.91} {\emph {\bibinfo {booktitle}
  {45th International Colloquium on Automata, Languages, and Programming (ICALP
  2018)}}},\ \bibinfo {series} {Leibniz International Proceedings in
  Informatics (LIPIcs)}, Vol.\ \bibinfo {volume} {107},\ \bibinfo {editor}
  {edited by\ \bibinfo {editor} {\bibfnamefont {I.}~\bibnamefont
  {Chatzigiannakis}}, \bibinfo {editor} {\bibfnamefont {C.}~\bibnamefont
  {Kaklamanis}}, \bibinfo {editor} {\bibfnamefont {D.}~\bibnamefont {Marx}},\
  and\ \bibinfo {editor} {\bibfnamefont {D.}~\bibnamefont {Sannella}}}\
  (\bibinfo  {publisher} {Schloss Dagstuhl--Leibniz-Zentrum fuer Informatik},\
  \bibinfo {address} {Dagstuhl, Germany},\ \bibinfo {year} {2018})\ pp.\
  \bibinfo {pages} {91:1--91:11}\BibitemShut {NoStop}%
\bibitem [{\citenamefont {Rains}\ \emph {et~al.}(1997)\citenamefont {Rains},
  \citenamefont {Hardin}, \citenamefont {Shor},\ and\ \citenamefont
  {Sloane}}]{nonstabilizer2}%
  \BibitemOpen
  \bibfield  {author} {\bibinfo {author} {\bibfnamefont {E.~M.}\ \bibnamefont
  {Rains}}, \bibinfo {author} {\bibfnamefont {R.}~\bibnamefont {Hardin}},
  \bibinfo {author} {\bibfnamefont {P.~W.}\ \bibnamefont {Shor}},\ and\
  \bibinfo {author} {\bibfnamefont {N.}~\bibnamefont {Sloane}},\ }\href@noop {}
  {\bibfield  {journal} {\bibinfo  {journal} {Physical Review Letters}\
  }\textbf {\bibinfo {volume} {79}},\ \bibinfo {pages} {953} (\bibinfo {year}
  {1997})}\BibitemShut {NoStop}%
\bibitem [{\citenamefont {Kitaev}(1997)}]{toric_kitaev}%
  \BibitemOpen
  \bibfield  {author} {\bibinfo {author} {\bibfnamefont {A.~Y.}\ \bibnamefont
  {Kitaev}},\ }\href@noop {} {\bibfield  {journal} {\bibinfo  {journal} {arXiv
  preprint quant-ph/9707021}\ } (\bibinfo {year} {1997})}\BibitemShut {NoStop}%
\bibitem [{\citenamefont {Fowler}\ \emph {et~al.}(2012)\citenamefont {Fowler},
  \citenamefont {Mariantoni}, \citenamefont {Martinis},\ and\ \citenamefont
  {Cleland}}]{Fowler}%
  \BibitemOpen
  \bibfield  {author} {\bibinfo {author} {\bibfnamefont {A.~G.}\ \bibnamefont
  {Fowler}}, \bibinfo {author} {\bibfnamefont {M.}~\bibnamefont {Mariantoni}},
  \bibinfo {author} {\bibfnamefont {J.~M.}\ \bibnamefont {Martinis}},\ and\
  \bibinfo {author} {\bibfnamefont {A.~N.}\ \bibnamefont {Cleland}},\ }\href
  {https://doi.org/10.1103/PhysRevA.86.032324} {\bibfield  {journal} {\bibinfo
  {journal} {Phys. Rev. A}\ }\textbf {\bibinfo {volume} {86}},\ \bibinfo
  {pages} {032324} (\bibinfo {year} {2012})}\BibitemShut {NoStop}%
\bibitem [{\citenamefont {Hastings}\ \emph {et~al.}(2020)\citenamefont
  {Hastings}, \citenamefont {Haah},\ and\ \citenamefont
  {O’Donnell}}]{RyanLD}%
  \BibitemOpen
  \bibfield  {author} {\bibinfo {author} {\bibfnamefont {M.~B.}\ \bibnamefont
  {Hastings}}, \bibinfo {author} {\bibfnamefont {J.}~\bibnamefont {Haah}},\
  and\ \bibinfo {author} {\bibfnamefont {R.}~\bibnamefont {O’Donnell}},\
  }\href@noop {} {\bibfield  {journal} {\bibinfo  {journal} {arXiv preprint
  arXiv:2009.03921}\ } (\bibinfo {year} {2020})}\BibitemShut {NoStop}%
\bibitem [{\citenamefont {Cross}\ \emph {et~al.}(2009)\citenamefont {Cross},
  \citenamefont {Smith}, \citenamefont {Smolin},\ and\ \citenamefont
  {Zeng}}]{CWS}%
  \BibitemOpen
  \bibfield  {author} {\bibinfo {author} {\bibfnamefont {A.}~\bibnamefont
  {Cross}}, \bibinfo {author} {\bibfnamefont {G.}~\bibnamefont {Smith}},
  \bibinfo {author} {\bibfnamefont {J.~A.}\ \bibnamefont {Smolin}},\ and\
  \bibinfo {author} {\bibfnamefont {B.}~\bibnamefont {Zeng}},\ }\href
  {https://doi.org/10.1109/TIT.2008.2008136} {\bibfield  {journal} {\bibinfo
  {journal} {IEEE Transactions on Information Theory}\ }\textbf {\bibinfo
  {volume} {55}},\ \bibinfo {pages} {433} (\bibinfo {year} {2009})}\BibitemShut
  {NoStop}%
\bibitem [{\citenamefont {Knill}\ and\ \citenamefont
  {Laflamme}(1997)}]{Laflamme}%
  \BibitemOpen
  \bibfield  {author} {\bibinfo {author} {\bibfnamefont {E.}~\bibnamefont
  {Knill}}\ and\ \bibinfo {author} {\bibfnamefont {R.}~\bibnamefont
  {Laflamme}},\ }\href@noop {} {\bibfield  {journal} {\bibinfo  {journal}
  {Physical Review A}\ }\textbf {\bibinfo {volume} {55}},\ \bibinfo {pages}
  {900} (\bibinfo {year} {1997})}\BibitemShut {NoStop}%
\bibitem [{\citenamefont {Looi}\ \emph {et~al.}(2008)\citenamefont {Looi},
  \citenamefont {Yu}, \citenamefont {Gheorghiu},\ and\ \citenamefont
  {Griffiths}}]{QuditGraph}%
  \BibitemOpen
  \bibfield  {author} {\bibinfo {author} {\bibfnamefont {S.~Y.}\ \bibnamefont
  {Looi}}, \bibinfo {author} {\bibfnamefont {L.}~\bibnamefont {Yu}}, \bibinfo
  {author} {\bibfnamefont {V.}~\bibnamefont {Gheorghiu}},\ and\ \bibinfo
  {author} {\bibfnamefont {R.~B.}\ \bibnamefont {Griffiths}},\ }\href@noop {}
  {\bibfield  {journal} {\bibinfo  {journal} {Physical Review A}\ }\textbf
  {\bibinfo {volume} {78}},\ \bibinfo {pages} {042303} (\bibinfo {year}
  {2008})}\BibitemShut {NoStop}%
\bibitem [{\citenamefont {Gallager}(1963)}]{Gallager}%
  \BibitemOpen
  \bibfield  {author} {\bibinfo {author} {\bibfnamefont {R.~G.}\ \bibnamefont
  {Gallager}},\ }\href@noop {} {\bibfield  {journal} {\bibinfo  {journal}
  {Available WWW: http://web. mit. edu/gallager/www/pages/ldpc. pdf}\ }
  (\bibinfo {year} {1963})}\BibitemShut {NoStop}%
\bibitem [{\citenamefont {Kuo}\ \emph {et~al.}(2021)\citenamefont {Kuo},
  \citenamefont {Chern}, \citenamefont {Lai} \emph {et~al.}}]{DS_decode}%
  \BibitemOpen
  \bibfield  {author} {\bibinfo {author} {\bibfnamefont {K.-Y.}\ \bibnamefont
  {Kuo}}, \bibinfo {author} {\bibfnamefont {I.}~\bibnamefont {Chern}}, \bibinfo
  {author} {\bibfnamefont {C.-Y.}\ \bibnamefont {Lai}}, \emph {et~al.},\
  }\href@noop {} {\bibfield  {journal} {\bibinfo  {journal} {arXiv preprint
  arXiv:2102.01984}\ } (\bibinfo {year} {2021})}\BibitemShut {NoStop}%
\bibitem [{\citenamefont {Kuo}\ and\ \citenamefont {Lai}(2020)}]{kuo}%
  \BibitemOpen
  \bibfield  {author} {\bibinfo {author} {\bibfnamefont {K.-Y.}\ \bibnamefont
  {Kuo}}\ and\ \bibinfo {author} {\bibfnamefont {C.-Y.}\ \bibnamefont {Lai}},\
  }\href@noop {} {\bibfield  {journal} {\bibinfo  {journal} {IEEE Journal on
  Selected Areas in Information Theory}\ }\textbf {\bibinfo {volume} {1}},\
  \bibinfo {pages} {487} (\bibinfo {year} {2020})}\BibitemShut {NoStop}%
\bibitem [{\citenamefont {Liu}\ and\ \citenamefont {Poulin}(2019)}]{BP2019}%
  \BibitemOpen
  \bibfield  {author} {\bibinfo {author} {\bibfnamefont {Y.-H.}\ \bibnamefont
  {Liu}}\ and\ \bibinfo {author} {\bibfnamefont {D.}~\bibnamefont {Poulin}},\
  }\href {https://doi.org/10.1103/PhysRevLett.122.200501} {\bibfield  {journal}
  {\bibinfo  {journal} {Phys. Rev. Lett.}\ }\textbf {\bibinfo {volume} {122}},\
  \bibinfo {pages} {200501} (\bibinfo {year} {2019})}\BibitemShut {NoStop}%
\bibitem [{\citenamefont {Huffman}\ and\ \citenamefont
  {Pless}(2010)}]{Huffaman}%
  \BibitemOpen
  \bibfield  {author} {\bibinfo {author} {\bibfnamefont {W.~C.}\ \bibnamefont
  {Huffman}}\ and\ \bibinfo {author} {\bibfnamefont {V.}~\bibnamefont
  {Pless}},\ }\href@noop {} {\emph {\bibinfo {title} {Fundamentals of
  error-correcting codes}}}\ (\bibinfo  {publisher} {Cambridge university
  press},\ \bibinfo {year} {2010})\BibitemShut {NoStop}%
\bibitem [{\citenamefont {Kitaev}\ \emph {et~al.}(2002)\citenamefont {Kitaev},
  \citenamefont {Shen}, \citenamefont {Vyalyi},\ and\ \citenamefont
  {Vyalyi}}]{Kitaev}%
  \BibitemOpen
  \bibfield  {author} {\bibinfo {author} {\bibfnamefont {A.~Y.}\ \bibnamefont
  {Kitaev}}, \bibinfo {author} {\bibfnamefont {A.}~\bibnamefont {Shen}},
  \bibinfo {author} {\bibfnamefont {M.~N.}\ \bibnamefont {Vyalyi}},\ and\
  \bibinfo {author} {\bibfnamefont {M.~N.}\ \bibnamefont {Vyalyi}},\
  }\href@noop {} {\emph {\bibinfo {title} {Classical and quantum
  computation}}},\ \bibinfo {number} {47}\ (\bibinfo  {publisher} {American
  Mathematical Soc.},\ \bibinfo {year} {2002})\BibitemShut {NoStop}%
\bibitem [{\citenamefont {Preskill}(1998)}]{Preskill_lecture}%
  \BibitemOpen
  \bibfield  {author} {\bibinfo {author} {\bibfnamefont {J.}~\bibnamefont
  {Preskill}},\ }\href@noop {} {\bibfield  {journal} {\bibinfo  {journal}
  {California Institute of Technology}\ }\textbf {\bibinfo {volume} {16}}
  (\bibinfo {year} {1998})}\BibitemShut {NoStop}%
\bibitem [{\citenamefont {Ashikhmin}\ \emph {et~al.}(2020)\citenamefont
  {Ashikhmin}, \citenamefont {Lai},\ and\ \citenamefont {Brun}}]{first_DS}%
  \BibitemOpen
  \bibfield  {author} {\bibinfo {author} {\bibfnamefont {A.}~\bibnamefont
  {Ashikhmin}}, \bibinfo {author} {\bibfnamefont {C.-Y.}\ \bibnamefont {Lai}},\
  and\ \bibinfo {author} {\bibfnamefont {T.~A.}\ \bibnamefont {Brun}},\
  }\href@noop {} {\bibfield  {journal} {\bibinfo  {journal} {IEEE Journal on
  Selected Areas in Communications}\ }\textbf {\bibinfo {volume} {38}},\
  \bibinfo {pages} {449} (\bibinfo {year} {2020})}\BibitemShut {NoStop}%
\bibitem [{\citenamefont {Ashikhmin}\ \emph {et~al.}(2016)\citenamefont
  {Ashikhmin}, \citenamefont {Lai},\ and\ \citenamefont {Brun}}]{DS_codes}%
  \BibitemOpen
  \bibfield  {author} {\bibinfo {author} {\bibfnamefont {A.}~\bibnamefont
  {Ashikhmin}}, \bibinfo {author} {\bibfnamefont {C.-Y.}\ \bibnamefont {Lai}},\
  and\ \bibinfo {author} {\bibfnamefont {T.~A.}\ \bibnamefont {Brun}},\ }in\
  \href@noop {} {\emph {\bibinfo {booktitle} {2016 IEEE International Symposium
  on Information Theory (ISIT)}}}\ (\bibinfo {organization} {IEEE},\ \bibinfo
  {year} {2016})\ pp.\ \bibinfo {pages} {2274--2278}\BibitemShut {NoStop}%
\bibitem [{\citenamefont {Schlingemann}\ and\ \citenamefont
  {Werner}(2001)}]{Schlingemann0}%
  \BibitemOpen
  \bibfield  {author} {\bibinfo {author} {\bibfnamefont {D.}~\bibnamefont
  {Schlingemann}}\ and\ \bibinfo {author} {\bibfnamefont {R.~F.}\ \bibnamefont
  {Werner}},\ }\href@noop {} {\bibfield  {journal} {\bibinfo  {journal}
  {Physical Review A}\ }\textbf {\bibinfo {volume} {65}},\ \bibinfo {pages}
  {012308} (\bibinfo {year} {2001})}\BibitemShut {NoStop}%
\bibitem [{\citenamefont {Schlingemann}(2001)}]{Schlingemann1}%
  \BibitemOpen
  \bibfield  {author} {\bibinfo {author} {\bibfnamefont {D.}~\bibnamefont
  {Schlingemann}},\ }\href@noop {} {\bibfield  {journal} {\bibinfo  {journal}
  {arXiv preprint quant-ph/0111080}\ } (\bibinfo {year} {2001})}\BibitemShut
  {NoStop}%
\bibitem [{\citenamefont {Luna}\ \emph {et~al.}(2014)\citenamefont {Luna},
  \citenamefont {Reid}, \citenamefont {De~Sanctis},\ and\ \citenamefont
  {Gheorghiu}}]{CAQEC}%
  \BibitemOpen
  \bibfield  {author} {\bibinfo {author} {\bibfnamefont {G.}~\bibnamefont
  {Luna}}, \bibinfo {author} {\bibfnamefont {S.}~\bibnamefont {Reid}}, \bibinfo
  {author} {\bibfnamefont {B.}~\bibnamefont {De~Sanctis}},\ and\ \bibinfo
  {author} {\bibfnamefont {V.}~\bibnamefont {Gheorghiu}},\ }\href@noop {}
  {\bibfield  {journal} {\bibinfo  {journal} {Discrete Mathematics, Algorithms
  and Applications}\ }\textbf {\bibinfo {volume} {6}},\ \bibinfo {pages}
  {1450054} (\bibinfo {year} {2014})}\BibitemShut {NoStop}%
\end{thebibliography}%
